\documentclass[letterpaper, 10pt, conference]{ieeeconf}  

\IEEEoverridecommandlockouts                              

\usepackage{amsmath}
\usepackage{amssymb}
\usepackage{amsfonts}
\usepackage{mathptmx} 
\usepackage{times}
\usepackage{graphicx}
\DeclareMathAlphabet{\bit}{OML}{cmm}{b}{it}

\newtheorem{thm}{Theorem}

\usepackage{datetime}
\usepackage{times} 
\usepackage{framed} 
\usepackage{graphicx} 
\usepackage{amsmath} 
\usepackage{amssymb}  

\def\fR{\mathfrak{R}}
\def\fH{\mathfrak{H}}

\def\<{\leqslant}           
\def\>{\geqslant}           

\def\d{\partial}
\def\wh{\widehat}

\def\Re{\mathrm{Re}}   

\def\cH{\mathcal{H}}   
\def\mR{\mathbb{R}}    
\def\mC{\mathbb{C}}    

\def\Tr{\mathrm{Tr}}       
\def\rT{\mathrm{T}}        

\def\bE{\mathbf{E}}    

\def\[[[{[\![\![}
\def\]]]{]\!]\!]}

\def\re{\mathrm{e}}        
\def\rd{\mathrm{d}}        

\def\bD{\mathbf{D}}

\def\x{\times}
\def\ox{\otimes}

\def\fF{{\mathfrak F}}

\def\fS{\mathfrak{S}}

\def\bH{\mathbf{H}}

\def\cI{\mathcal{I}}

\def\cE{\mathcal{E}}

\def\eps{\epsilon}

\def\ups{\upsilon}

\def\supp{\mathrm{supp}}    

\title{\LARGE \bf
Risk-sensitive Performance Criteria and Robustness of Quantum Systems with a Relative Entropy Description
of State Uncertainty$^*$}

\usepackage{datetime}

\author{Igor G. Vladimirov$^{\dagger}$, \quad Ian R. Petersen$^{\dagger}$, \quad Matthew R. James$^{\dagger}$
\thanks{$^*$This work is supported by the Air Force Office of Scientific Research (AFOSR) under agreement number FA2386-16-1-4065 and the Australian Research Council under grant DP180101805.}
\thanks{$^\dagger$Research School of Engineering, College of Engineering and Computer Science, Australian National University, ACT 2601, Canberra, Australia.
{\tt igor.g.vladimirov@gmail.com}, {\tt i.r.petersen@gmail.com}, {\tt matthew.james@anu.edu.au}.
}
}

\begin{document}
\maketitle
\thispagestyle{empty}

\begin{abstract}
This paper considers links between the original risk-sensitive performance criterion for quantum control systems and its recent quadratic-exponential counterpart.
We discuss a connection between the minimization of these cost functionals and robustness with respect to uncertainty in system-environment quantum states whose deviation from a nominal state is described in terms of the quantum relative entropy. These relations are similar to those in minimax LQG control for classical systems.
The results of the paper can be of use in providing a rational choice of the risk-sensitivity parameter in the context of robust quantum control with entropy theoretic quantification of statistical uncertainty in the system-field state.
\end{abstract}

\section{INTRODUCTION}

Linear quantum stochastic systems, or open quantum harmonic oscillators (OQHOs) \cite{GZ_2004}, are the principal models in linear quantum systems theory \cite{P_2017}. 
Open quantum systems  interact with the environment, which may include other quantum systems, external quantum fields  and classical systems, such as  measurement devices. In the case of OQHOs,  the energetics of this interaction is specified by quadratic Hamiltonians and linear system-field coupling operators with respect to dynamic variables,  which are 
operators on an underlying Hilbert space (for example, the quantum mechanical positions and momenta, or the annihilation and creation operators \cite{LL_1991,M_1998,S_1994}). The dynamics of these system variables are modelled by
Hudson-Parthasarathy  quantum stochastic differential equations (QSDEs) \cite{H_1991,HP_1984,P_1992,P_2015} which are linear in the case of OQHOs. The QSDEs are driven by quantum Wiener processes, which act on a symmetric Fock space \cite{P_1992,PS_1972} and represent external bosonic fields (such as quantized electromagnetic radiation).
Despite certain parallels between the quantum and classical linear SDEs, the noncommutative nature of the quantum variables (and the quantum probabilistic description \cite{H_2001,M_1995} of their statistical properties) leads to nontrivial quantum control and filtering problems for OQHOs (see, for example, \cite{B_1983,B_2010,BH_2006,BVJ_2007,EB_2005,J_2004,J_2005,JNP_2008,NJP_2009,VP_2013a,VP_2013b,WM_2010}).

These problems  are concerned with achieving a certain dynamic behavior for quantum plants of interest by using measurement-based feedback with classical controllers and filters or coherent (measurement-free) feedback   in the form of direct or field-mediated connections \cite{ZJ_2012} with other quantum systems, constituting a quantum feedback network  \cite{GJ_2009,JG_2010}. Quantum control and filtering applications aim to exploit quantum-mechanical resources for artificially engineered systems, for  example, using nonclassical light-matter interaction at atomic scales \cite{WM_2008} for quantum optical computing \cite{NC_2000}.

In addition to qualitative specifications (such as stability), the performance criteria employ optimality considerations in the form of cost functionals to be minimized. This approach is used in quantum linear quadratic Gaussian (LQG) control \cite{EB_2005,MP_2009,NJP_2009},  which, similarly to its classical predecessors \cite{AM_1989,KS_1972}, is concerned with  mean square values of the system variables (whose averaged values are also considered in quantum $\cH_{\infty}$-control settings \cite{JNP_2008}).   Quantum LQG control admits a
guaranteed-cost version \cite{SPJ_2007} with an extension to non-quadratic cost functionals for nonlinear QSDEs \cite{P_2014}. These settings take into account the presence of unmodelled uncertainties in the quantum dynamics  (such as deviations of state-space matrices from their nominal values or sector-bounded nonlinearities \cite{PUJ_2012}) and lead to upper bounds (and their minimization) for the worst-case values of the costs over the class of uncertain system dynamics.
The cost functionals in these approaches are formulated in terms of the second or higher-order moments of the system variables at one instant in time (or the integrals of the one-point moments over bounded time intervals in finite-horizon settings \cite{VP_2011b}).

The quantum risk-sensitive performance criterion \cite{J_2004,J_2005} (see also \cite{DDJW_2006,YB_2009}), which was used previously  for measurement-based quantum control and filtering problems,
 differs qualitatively from the cost functionals mentioned above.  Being a weighted mean square value of the time-ordered exponential for an operator differential equation (ODE), it involves higher-order moments of  the system variables at different instants. Since multi-point quantum states of noncommuting system variables do not reduce to classical joint probability distributions (even in the Gaussian case \cite{CH_1971,KRP_2010}),
 the quantum risk-sensitive cost is also different from its classical predecessors \cite{BV_1985,J_1973,W_1981}. The general structure of the classical risk-sensitive performance criteria has recently been used in a quadratic-exponential functional (QEF) \cite{VPJ_2017b}. This is an alternative version of the original quantum risk-sensitive cost and is organized as the exponential moment of the integral of a quadratic form of the system variables over a bounded time interval.  Both functionals impose an exponential penalty on the system variables, which is controlled by a risk-sensitivity parameter, and their computation and minimization (especially in the coherent quantum control setting mentioned above) are challenging problems. Their practical significance lies in potentially more conservative quantum dynamics secured by risk-sensitive controllers and filters. In fact, the minimization of the QEF cost leads to Cramer type upper bounds \cite{VPJ_2017b}  for the exponential decay in the tail distribution for the corresponding quadratic functional of the quantum system variables in the spirit of the large deviations theory \cite{DE_1997,S_1996}.

In the present paper, we obtain an integro-differential equation for the original quantum risk-sensitive cost functional and compare it with the evolution of the QEF studied in \cite{VPJ_2017b}.  This comparison shows that  the ODE for the original functional can be ``adjusted'' 
so as to reproduce the QEF, and thus the latter belongs to the wider class of quantum risk-sensitive costs of \cite{J_2004,J_2005}. As shown in \cite{VPJ_2017b}, the asymptotic growth rate of the QEF lends itself to successive computation (in terms of cumulants for multi-point Gaussian quantum states) for stable OQHOs driven by vacuum fields \cite{P_1992}. Although the QEF is associated with a specific system-field state, it gives rise to guaranteed upper bounds for the worst-case value of the corresponding quadratic cost when the actual state may differ from its nominal model. The ``size'' of this quantum statistical uncertainty is measured by the quantum relative entropy of the actual system-field state with respect to the nominal one.  The quantum robust performance estimates are based substantially on the results of \cite{J_2004,YB_2009} (and also references therein) and correspond to the connections between risk-sensitive control and minimax LQG control for classical systems with a relative entropy description of statistical uncertainty in the driving noise \cite{DJP_2000,P_2006,PJD_2000,PUS_2000}.

The paper is organised as follows.
Section~\ref{sec:funs} specifies the original quantum risk-sensitive cost functional and its quadratic-exponential counterpart.
Section~\ref{sec:links} compares the evolution equations for these functionals in the framework of OQHO dynamics.
Section~\ref{sec:worst} relates the QEF to the worst-case value of the quadratic cost over the class of uncertain system-field states with a given quantum relative entropy threshold.
Section~\ref{sec:upper} specifies the robust performance estimates for stable OQHOs.
Section~\ref{sec:conc} provides concluding remarks.

\section{QUANTUM RISK-SENSITIVE COST FUNCTIONALS}
\label{sec:funs}

We consider an open quantum system with an even number of dynamic variables $X_1, \ldots, X_n$ (for example, pairs of conjugate  quantum mechanical positions and momenta \cite{S_1994}). These system variables  are time-varying self-adjoint operators on a complex separable Hilbert space $\fH$ (or a dense domain thereof). For what follows, they are assembled into the vector
\begin{equation}
\label{X}
    X:=
    \begin{bmatrix}
        X_1\\
        \vdots\\
        X_n
    \end{bmatrix}
\end{equation}
(vectors are organised as columns unless specified otherwise, and the time argument is often omitted for brevity).  The statistical properties of the system variables depend on a density operator (quantum state) $\rho$ which is a positive semi-definite self-adjoint operator on $\fH$ with unit trace $\Tr \rho = 1$. The quantum state gives rise to the quantum expectation
\begin{equation}
\label{bE}
    \bE \xi := \Tr(\rho \xi)
\end{equation}
of a quantum variable $\xi$ on the underlying space $\fH$. In accordance with the fact that a self-adjoint  quantum variable $\xi = \xi^\dagger$ represents a real-valued physical quantity, its averaging leads to a real mean value: $\bE \xi \in \mR$. Furthermore, if $\xi$ is positive semi-definite, then $\bE \xi = \Tr (\eta \eta^{\dagger}) \> 0$, where $\eta := \sqrt{\rho}\sqrt{\xi}$ is an auxiliary quantum variable (not necessarily self-adjoint)  which uses positive semi-definite square roots of such operators.

Moderate mean square values $\bE (X_k^2)$ of the system variables in (\ref{X}) may be required for well-posedness  of the system (for example, from energy considerations or in order to keep the system in the range where its linearised model is satisfactory).  Their large values can be penalised by cost functionals similar to those in classical LQG and risk-sensitive control \cite{AM_1989,BV_1985,J_1973,W_1981}. However, the quantum setting is complicated by the fact that the system variables do not   commute with each other, and this is taken into account in the quantum counterparts of the conventional performance criteria.

The quantum risk-sensitive cost functional, proposed in \cite{J_2004,J_2005}, employs an auxiliary quantum process which is defined as the solution of the operator differential equation
\begin{equation}
\label{Rdot}
    \dot{R}_{\theta}(t) = \frac{\theta }{2}C(t) R_{\theta}(t),
    \qquad
    t\>0.
\end{equation}
Here, $\dot{(\, )}:= \d_t(\cdot)$ denotes the time derivative, $\theta\>0$ is a parameter, whose role is clarified below, and
$C(t)$ is a time-dependent positive semi-definite self-adjoint quantum variable which can be a function of the current system variables (or their past history over the time interval $[0,t]$ in a more general case). Assuming that the initial condition for the equation (\ref{Rdot}) is the identity operator (that is, $R_{\theta}(0) = \cI_{\fH}$), its fundamental solution is given by
the (leftward)  time ordered exponential
\begin{equation}
\label{Rt}
    R_{\theta}(t)
    =
    \mathop{\overleftarrow{\exp}}
    \Big(
    \frac{\theta}{2}
    \int_0^t
    C(s)
    \rd s
    \Big).
\end{equation}
Although $R_{\theta}(t)$ involves exponentiation of the system operators and its averaging imposes a penalty on their exponential moments, $R_{\theta}(t)$ is, in general,  a non-Hermitian operator with a complex mean value. Since real-valued costs are convenient for comparison (the set of reals is linearly ordered), a weighted mean square value of the quantum process $R_{\theta}$ is considered:
\begin{equation}
\label{Et}
    E_{\theta}(t):= \bE
    \big(
        R_{\theta}(t)^{\dagger} \re^{\theta T(t)}R_{\theta}(t)
    \big),
\end{equation}
cf. \cite[Eqs. (19)--(21)]{J_2005}. Here,
$T(t)$ is a positive semi-definite self-adjoint operator whose role is similar to that of the terminal cost (on the time interval $[0,t]$) in classical control problems.
For simplicity, this additional cost will not be included in what follows, so that  $T(t) = 0$, and (\ref{Et}) takes the form
\begin{equation}
\label{Et0}
    E_{\theta}(t)= \bE
    \big(
        R_{\theta}(t)^{\dagger}
        R_{\theta}(t)
    \big).
\end{equation}
If
the quantum variables $C(s)$ commuted  with each other for all $0\< s\< t$,  then (\ref{Et}) would reduce to
\begin{equation}
\label{Eclass}
    E_{\theta}(t)
    =
    \bE
    \re^{
        \theta
            \int_0^t C(s)\rd s
    },
\end{equation}
which is organised as the classical exponential-of-integral performance criteria \cite{BV_1985,J_1973,W_1981}, with $\theta$ being the risk-sensitivity parameter. The latter should be small enough in order for the exponential moments to be finite (otherwise, $E_\theta(t)=+\infty$).  Moreover, in the noncommutative quantum setting, the right-hand side of (\ref{Eclass}) provides an alternative to the original quantum risk-sensitive cost functional defined by (\ref{Rt}) and (\ref{Et0}). Its quadratic-exponential counterpart, considered recently in \cite{VPJ_2017b}, is given by
\begin{equation}
\label{QEF}
    \Xi_{\theta}(t)
    :=
    \bE \re^{\theta\varphi(t)},
\end{equation}
where $\varphi$ is a quantum process defined for any time $t\> 0$ as the integral of a  quadratic function $\psi$  of the system variables in (\ref{X}) over the time interval $[0,t]$:
\begin{align}
\label{phi}
    \varphi(t)
    & :=
    \int_0^t
    \psi(s)
    \rd s,\\
\label{psi}
    \psi(s)
    & :=
    X(s)^{\rT} \Pi X(s).
\end{align}
Here, the transpose $(\cdot)^\rT$  applies to vectors and matrices of operators as if their entries were scalars, and $\Pi$ is a real positive semi-definite symmetric matrix of order $n$ (the dependence of $\Xi_{\theta}(t)$ on $\Pi$ is omitted for brevity). Accordingly, $\varphi(t)$ and $\psi(t)$ are positive semi-definite self-adjoint operators on the underlying Hilbert space $\fH$, which follows from the representation
\begin{equation}
\label{psizeta}
    \psi
    =
    \zeta^{\rT}\zeta
    =
    \sum_{k=1}^n
    \zeta_k^2
\end{equation}
in terms of the auxiliary self-adjoint quantum variables constituting the vector
\begin{equation}
\label{zeta}
    \zeta
    :=
    \begin{bmatrix}
        \zeta_1\\
        \vdots\\
        \zeta_n
    \end{bmatrix}
    :=
    \sqrt{\Pi} X.
\end{equation}
The risk-sensitive cost $E_{\theta}$, given by (\ref{Rt}) and (\ref{Et0}), and its quadratic-exponential counterpart $\Xi_{\theta}$ in (\ref{QEF})--(\ref{psi}) have a similar asymptotic behaviour for small values of $\theta$ in the sense that
\begin{align}
\label{Easy}
    \lim_{\theta\to 0+ }
    \Big(
        \frac{1}{\theta}
        \ln E_{\theta}(t)
    \Big)
    & =
    \int_0^t \bE C(s)\rd s,\\
\label{Xiasy}
    \lim_{\theta\to 0+ }
    \Big(
        \frac{1}{\theta}
        \ln \Xi_{\theta}(t)
    \Big)
    & =
    \int_0^t \bE \psi(s)\rd s
\end{align}
under suitable integrability conditions. Moreover, they
are identical for any $\theta\> 0$ in the classical case if $C=\psi$. However,  these performance indices are different in the noncommutative quantum setting because of the discrepancy between the time-ordered exponential and the usual operator exponential. In particular, the QEF $\Xi_{\theta}(t)$ in (\ref{QEF}) is the moment-generating function (and $\ln \Xi_{\theta}(t)$ in (\ref{Xiasy}) is the cumulant-generating function) for the classical probability distribution (the averaged spectral measure \cite{H_2001})  of the self-adjoint quantum variable $\varphi(t)$ at a given instant $t\>0$. In contrast to $\Xi_{\theta}(t)$, the quantity $E_\theta(t)$ in (\ref{Et0}) (whose logarithm is present in (\ref{Easy})),  does not lend itself (as a function of $\theta$) to a similar  association with the probability distribution of a single quantum variable.

\section{RELATIONS BETWEEN THE FUNCTIONALS}
\label{sec:links}

We will now consider the time evolution of the original risk-sensitive cost functional. To this end, we associate the following modified density operator
\begin{equation}
\label{rt}
    r_{\theta, t}
    :=
    \frac{1}{E_{\theta}(t)}
    R_{\theta}(t)\rho R_{\theta}(t)^{\dagger}
\end{equation}
with the quantum state $\rho$ in (\ref{bE}).
The property that $r_{\theta, t}$ is a density operator is ensured by its self-adjointness, positive semi-definiteness and the unit trace property
\begin{align*}
    \Tr r_{\theta, t}
     & =
    \frac{1}{E_{\theta}(t)}
    \Tr \big(R_{\theta}(t)\rho R_{\theta}(t)^{\dagger}\big)\\
     & =
    \frac{1}{E_{\theta}(t)}
    \Tr \big(\rho R_{\theta}(t)^{\dagger}R_{\theta}(t)\big)\\
    & =
    \frac{1}{E_{\theta}(t)}
    \bE \big(R_{\theta}(t)^{\dagger}R_{\theta}(t)\big)=1
\end{align*}
in view of (\ref{Et0}).
The quantum expectation over the modified density operator $r_{\theta, t}$ in (\ref{rt}) takes the form
\begin{align}
\nonumber
  \cE_{\theta, t}\xi
  & :=
  \Tr(r_{\theta,t}\xi)\\
\nonumber
  & =
    \frac{1}{E_{\theta}(t)}
    \Tr
    \big(
        R_{\theta}(t) \rho R_{\theta}(t)^{\dagger}
        \xi
    \big)\\
\label{cEt}
    & =
    \frac{1}{E_{\theta}(t)}
    \bE
    \big(
    R_{\theta}(t)^{\dagger}
    \xi
    R_{\theta}(t)
    \big).
\end{align}
If $\theta=0$, then (\ref{Rdot}) implies that the exponential in (\ref{Rt}) reduces to  $R_0(t) = \cI_{\fH}$,  and hence, $E_0(t) = 1$ in view of (\ref{Et0}). In this limiting case, the definitions (\ref{rt}) and (\ref{cEt}) yield the original quantum state $r_{0,t}=\rho$ and the original expectation $\cE_{0,t}=\bE$ in (\ref{bE}).

\begin{thm}
\label{th:Edot}
The quantum risk-sensitive cost functional $E_{\theta}(t)$, given by (\ref{Rt}) and  (\ref{Et0}), satisfies the integro-differential equation
\begin{equation}
\label{Edot}
    (\ln E_{\theta}(t))^{^\centerdot}
  =
  \theta
  \cE_{\theta,t}
    C(t),
\end{equation}
where $\cE_{\theta,t}$ is the modified quantum expectation given by (\ref{cEt}).
\hfill$\square$
\end{thm}
\begin{proof}
By using (\ref{Rdot}) and self-adjointness of the operator  $C(t)$, it follows that
\begin{align}
\nonumber
        (R^{\dagger}R)^{^\centerdot}
        & =
        \dot{R}^\dagger R + R^\dagger \dot{R}\\
\nonumber
        & =
        \frac{\theta}{2}
        \big(
        (C R)^\dagger R + R^\dagger CR
        \big)\\
\label{RRdot}
    & =
      \theta R^\dagger C R,
\end{align}
where the subscript $\theta$ is also omitted for brevity.
Now, the time differentiation of (\ref{Et0}) and the averaging of (\ref{RRdot}) lead to
\begin{align}
\nonumber
    \dot{E}_{\theta}(t)
    & =
    \bE (R_{\theta}(t)^{\dagger}R_{\theta}(t))^{^\centerdot}\\
\label{Edot1}
    & =
    \theta \bE \big(R_{\theta}(t)^\dagger C(t) R_{\theta}(t)\big).
\end{align}
Division of both sides of (\ref{Edot1}) by $E_{\theta}(t)$ allows the logarithmic  time derivative of the cost functional to be represented in the form
\begin{align*}
    (\ln E_{\theta}(t))^{^\centerdot}
    & =
    \frac{\theta}{E_{\theta,t}}
     \bE \big(R_{\theta}(t)^\dagger C(t) R_{\theta}(t)\big)\\
     & =
     \theta \cE_{\theta, t} C(t),
\end{align*}
whose right-hand side is related to the modified expectation (\ref{cEt}), thus establishing (\ref{Edot}).
\end{proof}

Note that Theorem~\ref{th:Edot}, which is a direct corollary of the ODE (\ref{Rdot}), does not use specific assumptions on the dynamics and the commutation structure of the system variables themselves. However, the time evolution of the QEF \cite{VPJ_2017b} depends on a more detailed description of the quantum dynamics. To this end, we assume that the system variables satisfy the Weyl CCRs whose infinitesimal Heisenberg form is
\begin{align}
\nonumber
    [X, X^{\rT}]
      & :=    ([X_j,X_k])_{1\< j,k\< n}\\
\nonumber
     & =    XX^{\rT}- (XX^{\rT})^{\rT}\\
\label{Theta}
     & =
     2i \Theta \ox \cI_{\fH}
\end{align}
(on a dense domain in $\fH$). Here, $i:= \sqrt{-1}$ is the imaginary unit, $[\alpha,\beta]:= \alpha \beta-\beta\alpha$ is the commutator of linear operators, $\Theta:= (\theta_{jk})_{1\< j,k\< n}$ is a nonsingular real antisymmetric matrix of order $n$, and $\ox$ is the tensor product. In what follows,  the matrix $\Theta \ox \cI_{\fH} = (\theta_{jk}\cI_{\fH})_{1\< j,k\< n}$ in (\ref{Theta}) is identified with $\Theta$. Furthermore, the system variables are assumed to be governed by a linear quantum stochastic differential equation (QSDE)
\begin{equation}
\label{dX}
    \rd X
    =
    A X \rd t+ B \rd W,
\end{equation}
where $A\in \mR^{n\x n}$ and $B \in \mR^{n\x m}$ are constant matrices whose structure is clarified below.
This QSDE is driven by the vector
$$
  W:=
  \begin{bmatrix}
        W_1\\
        \vdots\\
        W_m
  \end{bmatrix}
$$
of an even number $m$ of quantum Wiener processes $W_1, \ldots, W_m$ which are time-varying self-adjoint operators on a symmetric Fock space $\fF$ \cite{P_1992,PS_1972}. These operators  represent the external bosonic fields and have a complex positive semi-definite Hermitian Ito matrix $\Omega \in \mC^{m\x m}$:
\begin{equation}
\label{WW}
    \rd W \rd W^{\rT}
    =
    \Omega \rd t,
    \qquad
    \Omega := I_m + iJ,
\end{equation}
where $I_m$ is the identity matrix of order $m$. The imaginary part
\begin{equation}
\label{J}
        J
        :=
       \begin{bmatrix}
           0 & I_{m/2}\\
           -I_{m/2} & 0
       \end{bmatrix}
\end{equation}
of $\Omega$ is an orthogonal real antisymmetric matrix of order $m$ (so that $J^2=-I_m$), which specifies CCRs for the quantum Wiener processes as $[\rd W, \rd W^{\rT}] = 2iJ\rd t$. The state-space matrices $A$ and $B$ in (\ref{dX}) are not arbitrary and satisfy the algebraic equation
\begin{equation}
\label{PR}
  A \Theta + \Theta A^{\rT} + BJB^{\rT} = 0,
\end{equation}
which is part of the physical realizability conditions \cite{JNP_2008,SP_2012} and is closely related to the preservation of the CCRs (\ref{Theta}) in time. The matrix pairs $(A,B)$ satisfying (\ref{PR}) are parameterized as
\begin{equation}
\label{AB}
  A = 2\Theta (K + M^\rT J M),
  \quad
  B = 2\Theta M^{\rT}
\end{equation}
in terms of matrices $K = K^{\rT} \in \mR^{n\x n}$ and $M \in \mR^{m\x n}$, which specify the system Hamiltonian $\frac{1}{2} X^\rT K X$ and the vector $MX$ of system-field coupling operators.
The relations (\ref{Theta})--(\ref{AB}) describe an open quantum harmonic oscillator (OQHO) whose  internal dynamics is affected by the interaction with the external bosonic fields. Accordingly,  the tensor-product system-field Hilbert space is given by
\begin{equation}
\label{fH}
    \fH := \fH_0 \ox \fF,
\end{equation}
where $\fH_0$ is a Hilbert space for the action of the initial system variables $X_1(0), \ldots, X_n(0)$. Also, the system-field density operator $\rho$ in (\ref{bE}) is also
assumed to have a tensor-product structure:
\begin{equation}
\label{rho}
  \rho
  :=
  \rho_0 \ox \ups,
\end{equation}
where $\rho_0$ is the initial system state on $\fH_0$ in (\ref{fH}), and the fields are in the vacuum state $\ups$ \cite{HP_1984,P_1992}.
The CCRs (\ref{Theta}), which are concerned with one point in time, extend to different moments of time as
\begin{equation}
\label{XsXtcomm}
    [X(s), X(t)^{\rT}]
    =
    2i\Lambda(s-t),
    \qquad
    s,t\> 0,
\end{equation}
where
\begin{equation}
\label{Lambda}
    \Lambda(\tau)
    =
    \left\{
    \begin{matrix}
    \re^{\tau A}\Theta & {\rm if}\ & \tau\> 0\\
    \Theta \re^{-\tau A^{\rT}} & {\rm if}\ & \tau< 0\\
    \end{matrix}
    \right.
\end{equation}
is the two-point CCR matrix of the system variables, with $\Lambda(0) = \Theta$.
Now, the formulation of the following theorem also uses a time-varying change of the density operator in (\ref{bE}):
\begin{equation}
\label{rhot}
    \rho_{\theta, t}
    :=
    \frac{1}{\Xi_{\theta}(t)}
    S_\theta(t)
    \rho
    S_\theta(t),
\end{equation}
where
\begin{equation}
\label{S}
  S_\theta(t)
  :=
  \re^{\frac{\theta}{2}\varphi(t)}
\end{equation}
is the positive definite self-adjoint square root of the operator $\re^{\theta \varphi(t)}$.
This quantum state transformation resembles (\ref{rt}), except that the left and right factors in (\ref{rhot}) are identical self-adjoint operators, and the normalization employs the QEF from (\ref{QEF}) instead of (\ref{Et0}).
The quantum expectation over the density operator $\rho_{\theta, t}$ in (\ref{rhot}) is computed as
\begin{equation}
\label{bEt}
  \bE_{\theta, t}\xi
  :=
    \frac{1}{\Xi_{\theta}(t)}
    \bE
    \big(
    S_\theta(t)
    \xi
    S_\theta(t)
    \big).
\end{equation}
Similarly to (\ref{rt}) and (\ref{cEt}), the  definitions (\ref{rhot}) and (\ref{bEt}) reproduce the original density operator $\rho_{0,t}=\rho$ and the original expectation $\bE_{0,t}=\bE$ in the case
$\theta=0$.  We will now provide a relevant part of \cite[Theorem~1]{VPJ_2017b}.

\begin{thm}
\label{th:ODE}
The QEF $\Xi_{\theta}(t)$ in (\ref{QEF}), associated with the OQHO in (\ref{Theta})--(\ref{AB}), satisfies the integro-differential equation
\begin{equation}
\label{dotQEF}
    (\ln \Xi_{\theta}(t))^{^\centerdot}
    =
    \theta
    \bE_{\theta,t}
    \Psi_{\theta}(t),
\end{equation}
where $\bE_{\theta,t}$ is the modified quantum expectation given by  (\ref{bEt}). Here,
\begin{align}
\nonumber
  \Psi_{\theta}(t)
  := &
  \psi(t)\\
\nonumber
   & +
    \frac{\theta}{2}
    \left(
        \Re
        \Big(
            X(t)^{\rT}
            \int_0^t
            \alpha_{\theta,t}(\sigma)
            X(\sigma)
            \rd \sigma
        \Big)
    \right.\\
\label{Psi}
        &+
    \left.
        \int_{[0,t]^2}
        X(\sigma)^{\rT}
        \beta_{\theta,t}(\sigma,\tau)
        X(\tau)
        \rd \sigma
        \rd \tau
    \right)
\end{align}
is a time-varying self-adjoint operator which is a quadratic function of the past history of the system variables, with $\psi$ given by  (\ref{psi}),
and the real part $\Re(\cdot)$ extended to operators as $\Re \xi:= \frac{1}{2}(\xi+\xi^{\dagger})$.
Also, the functions $\alpha_{\theta, t}: [0,t]\to \mR^{n\x n}$ and $\beta_{\theta,t}: [0,t]^2\to \mR^{n\x n}$ are related to the two-point CCR matrix $\Lambda$ of the system variables in (\ref{XsXtcomm}) and (\ref{Lambda}) as described in \cite[Theorem~1 and Lemma~2]{VPJ_2017b}.
\hfill$\square$
\end{thm}

The proof of Theorem~\ref{th:ODE}, given in \cite{VPJ_2017b}, is based substantially on the relation
\begin{equation}
\label{ephidot}
    \big(
        \re^{\theta \varphi(t)}
    \big)^{^\centerdot}
    =
    \theta
    S_\theta(t)
    \Psi_\theta(t)
    S_\theta(t),
\end{equation}
which is obtained by using (\ref{S}) and the fact that quadratic functions of quantum variables, satisfying CCRs,  form a Lie algebra with respect to the commutator (see, for example, \cite[Appendix A]{VPJ_2017b} and references therein). This plays its role in combination with the property that both $\varphi(t)$ in (\ref{phi}) and its time derivative  $\psi(t) = \dot{\varphi}(t)$ in (\ref{psi}) are quadratic functions of the system variables (on the time interval $[0,t]$) which satisfy the CCRs (\ref{XsXtcomm}).

Since, in general, $[\varphi(t),\psi(t)]\ne 0$ (that is,  $\varphi(t)$ and $\psi(t)$ do not commute), the self-adjoint quantum  variable $\Psi_{\theta}(t)$ in (\ref{dotQEF}) and (\ref{Psi}) is not necessarily positive semi-definite in contrast to $C(t)$ in (\ref{Edot}).
Despite these discrepancies, the quantum risk-sensitive and quadratic-exponential functionals can be reconciled  by an appropriate choice of the process $C$ which drives (\ref{Rdot}) and (\ref{Rt}) (with the condition $C(t)\succcurlyeq 0$ being omitted).

\begin{thm}
\label{th:ort}
Suppose the quantum process $C$ in (\ref{Rdot}) is related to $\Psi_{\theta}$ in (\ref{Psi}) by
\begin{equation}
\label{CV}
    C(t) = V_\theta(t)\Psi_\theta(t) V_\theta(t)^\dagger,
\end{equation}
for all $t\>0$, where
\begin{equation}
\label{V}
    V_\theta(t)
    :=
    R_\theta(t)S_{\theta}(t)^{-1}
\end{equation}
is a time-varying operator defined in terms of (\ref{Rt})  and (\ref{S}). Then the original quantum risk-sensitive cost functional $E_{\theta}(t)$ in (\ref{Et0}) reproduces its quadratic-exponential counterpart $\Xi_{\theta}(t)$ in (\ref{QEF}), with the operator $V_\theta(t)$ remaining unitary for all $t\>0$.
\hfill$\square$
\end{thm}
\begin{proof}
A sufficient condition for $E_{\theta}(t)=\Xi_\theta(t)$ is provided by
\begin{equation}
\label{RRexp}
    R_\theta(t)^\dagger R_\theta(t)
    =
    \re^{\theta \varphi(t)}
\end{equation}
(indeed, equal quantum variables have the same mean value).
In view of the exponential structure of the operators $R_\theta(t)$ and $S_\theta(t)$ in (\ref{Rt}) and (\ref{S}), the relation (\ref{RRexp}) holds if and only if the operator
$V_\theta(t)$ in (\ref{V}) is unitary, which is equivalent to
\begin{align}
\nonumber
    U_\theta(t)
    & :=
    V_\theta(t)^\dagger V_\theta(t)\\
\label{UV}
    & =
    S_{\theta}(t)^{-1} R_\theta(t)^\dagger R_\theta(t)S_{\theta}(t)^{-1}
\end{align}
being the identity operator. The unitarity holds trivially for the initial value $V_\theta(0) = \cI_\fH$ since $R_\theta(0) = S_\theta(0)=\cI_\fH$. The process $C$ can be organized so as to propagate this property over time. To this end, the evolution of $U_\theta(t)$ for $t\>0$ is obtained by differentiating  both sides of (\ref{UV}) with respect to time and using (\ref{RRdot}) together with the identity $(S^{-1})^{^\centerdot} = -S^{-1}\dot{S}S^{-1}$:
\begin{align}
\nonumber
    \dot{U}
    = &
    S^{-1} (R^\dagger R)^{^\centerdot} S^{-1}    \\
\nonumber
    & + (S^{-1})^{^\centerdot} R^\dagger R S^{-1}
    +
    S^{-1} R^\dagger R (S^{-1})^{^\centerdot} \\
\nonumber
    = &
    \theta S^{-1} R^\dagger C R S^{-1}\\
\nonumber
    & -S^{-1} \dot{S}S^{-1}R^\dagger R S^{-1}
     -S^{-1} R^\dagger R S^{-1} \dot{S}S^{-1}\\
\label{Udot}
    = &
    \theta V^\dagger C V
    -S^{-1} \dot{S} U
    -U \dot{S}S^{-1}.
\end{align}
Here, the subscript $\theta$ is omitted for brevity, and  use is also made of  (\ref{V}) (along with self-adjointness of the operator $S_\theta(t)$ in (\ref{S})). The unitarity of $V$ is preserved in time if the identity operator $U=\cI_{\fH}$ satisfies (\ref{Udot}), which takes the form
\begin{equation}
\label{Idot}
    \theta V^\dagger C V
    -
    S^{-1} \dot{S}
    -
    \dot{S}S^{-1} = 0.
\end{equation}
By combining the identity $(S^2)^{^\centerdot}= \dot{S}S + S\dot{S}$ with (\ref{ephidot}), it follows that
\begin{equation}
\label{Sdot}
    S^{-1} \dot{S}
    +
    \dot{S}S^{-1}
    =
    S^{-1}
    (S^2)^{^\centerdot}
    S^{-1}
    =
    \theta \Psi.
\end{equation}
Substitution of (\ref{Sdot}) into (\ref{Idot}) transforms the latter equation to
\begin{equation}
\label{Idot1}
    \theta (V^\dagger C V-\Psi) = 0
\end{equation}
Therefore, if $C$ is given by (\ref{CV}), then (\ref{Idot1}) holds due to $V$ in (\ref{V}) being unitary, whereby this property is preserved in time, and so also is (\ref{RRexp}).
\end{proof}

The operator $C(t)$, specified by Theorem~\ref{th:ort},   is unitarily equivalent to $\Psi_\theta(t)$.  Furthermore, by substituting (\ref{CV}) and (\ref{V}) into (\ref{Rdot}), it follows that the corresponding process $R_\theta$ is governed by
\begin{align}
\nonumber
    \dot{R}_{\theta}(t)
    & =
    \frac{\theta }{2}
    \overbrace{R_\theta(t)S_{\theta}(t)^{-1} \Psi_\theta(t) S_{\theta}(t)^{-1} R_\theta(t)^\dagger}^{C(t)} R_{\theta}(t)\\
\label{Rdot1}
    & = \frac{\theta }{2}
    R_\theta(t)S_{\theta}(t)^{-1} \Psi_\theta(t) S_{\theta}(t),
\end{align}
which is a linear ODE despite the quadratic dependence of $C$ on $R_\theta$. This reduction in (\ref{Rdot1}) holds due to the relation (\ref{RRexp}) (in the framework of Theorem~\ref{th:ort}) and the square root property $\re^{\theta \varphi(t)} = S_\theta(t)^2$.

\section{ENTROPY THEORETIC UNCERTAINTIES IN QUANTUM STATES}
\label{sec:worst}

For any time horizon $t>0$, let $\rho_t$ denote the system-field quantum state over the time interval $[0,t]$. In accordance with (\ref{rho}), the \emph{nominal state} is organised as
\begin{equation}
\label{rho0}
    \wh{\rho}_t
    :=
    \rho_0\ox \ups_t,
\end{equation}
where $\ups_t$ is the vacuum field state on the Fock subspace $\fF_t$ associated with $[0,t]$. In general, the actual state $\rho_t$ is not known precisely and may differ from $\wh{\rho}_t$, with both acting on the subspace
\begin{equation}
\label{fHt}
    \fH_t
    :=
    \fH_0 \ox \fF_t
\end{equation}
of the Hilbert space $\fH$ in (\ref{fH}).  This deviation from the nominal quantum state can be described by the quantum relative entropy \cite{OW_2010}
\begin{align}
\nonumber
  \bD(\rho_t \| \wh{\rho}_t)
  & :=
  \Tr
  (\rho_t (\ln \rho_t - \ln \wh{\rho}_t))  \\
\label{bD}
    & =
  - \bH(\rho_t)
  -\bE\ln \wh{\rho}_t,
\end{align}
where $\bH(\rho_t):= -\Tr(\rho_t \ln \rho_t) = -\bE \ln \rho_t$ is the von Neumann entropy \cite{NC_2000} of $\rho_t$, and the expectation is over the actual state; cf. \cite[Eq. (7)]{YB_2009}. It is assumed that the supports of the density operators satisfy the inclusion $\supp \rho_t\subset \supp \wh{\rho}_t$ (as subspaces of $\fH_t$ in (\ref{fHt})). Similarly to its classical counterpart \cite{CT_2006}, the quantity (\ref{bD}) is always nonnegative and vanishes only if $\rho_t = \wh{\rho}_t$.  Since it is the actual density operator (rather than its nominal model) that determines physically meaningful statistical properties of quantum variables, the discrepancy between $\rho_t$ and $\wh{\rho}_t$ can be interpreted as a quantum statistical uncertainty. The corresponding class of uncertain system-field quantum states $\rho_t$ can be represented as
\begin{equation}
\label{class}
    \fR_{t,\eps }
    :=
  \big\{
    \rho_t:\
    \bD(\rho_t \| \wh{\rho}_t) \< \eps t
  \big\},
\end{equation}
where $\eps$ is a given nonnegative parameter which bounds the quantum relative entropy growth rate in (\ref{bD}). In particular, the case $\eps = 0$ corresponds to the absence of uncertainty when the set in (\ref{class})  is a singleton consisting of the nominal state: $\fR_{t,0} = \{\wh{\rho}_t\}$. At the other extreme, for large values of $\eps$, some elements $\rho_t \in \fR_{t,\eps}$ of the uncertainty class can lead to substantially higher values of the cost functionals than those predicted by the nominal model. In application to the quadratic functional of the system variables in (\ref{phi}), the following theorem relates the worst-case value of the corresponding quadratic cost to the QEF
\begin{equation}
\label{QEFnom}
    \wh{\Xi}_\theta(t)
    :=
    \wh{\bE} \re^{\theta \varphi(t)}
    =
    \Tr(\wh{\rho}_t \re^{\theta \varphi(t)}).
\end{equation}
Note that, despite its similarity to (\ref{QEF}), the definition (\ref{QEFnom}) employs the expectation $\wh{\bE}(\cdot)$ over the nominal system-field state $\wh{\rho}_t$ in (\ref{rho0}) which can now be different from the actual state $\rho_t$.

\begin{thm}
\label{th:worst}
Suppose the actual system-field state $\rho_t$ (over the time interval $[0,t]$) belongs to the uncertainty class $\fR_{t,\eps}$ in (\ref{class}) for all $t\>0$. Then, for any given $\eps\>0$,  the growth rate for the worst-case value of the quadratic cost $\varphi(t)$ in (\ref{phi}) admits an upper bound
\begin{equation}
\label{supEphi}
  \limsup_{t\to+\infty}
  \Big(
  \frac{1}{t}
  \sup_{\rho_t \in \fR_{t,\eps}}
  \bE \varphi(t)
  \Big)
  \<
  \inf_{\theta >0}
    \frac{\eps + \gamma(\theta)}{\theta}.
\end{equation}
Here,
\begin{equation}
\label{gamma}
  \gamma(\theta)
  :=
  \limsup_{t\to +\infty}
  \Big(
  \frac{1}{t}
  \ln \wh{\Xi}_\theta(t)
  \Big)
\end{equation}
quantifies the upper growth rate for the nominal QEF in (\ref{QEFnom}).
\hfill$\square$
\end{thm}
\begin{proof}
Application of \cite[Lemma 2.1]{YB_2009} and its corollary \cite[Eq. (9)]{YB_2009} (which are based on the Golden-Thompson inequality $\Tr(\re^{\xi+\eta}) \< \Tr (\re^\xi \re^\eta)$ for the exponentials of self-adjoint operators \cite{G_1965,OP_1993,T_1965}) to the self-adjoint quantum variable $\theta \varphi(t)$ yields
\begin{equation}
\label{GT}
  \theta \bE \varphi(t)
  \<
  \bD(
    \rho_t
    \|
    \wh{\rho}_t
  )
  +
  \ln \wh{\Xi}_\theta(t),
\end{equation}
where $\wh{\Xi}_\theta(t)$ is given by 
(\ref{QEFnom}). In contrast to the duality relation, which is used in classical minimax LQG control with a relative entropy description of statistical uncertainty \cite{DJP_2000,P_2006,PJD_2000,PUS_2000}, the equality in the quantum counterpart (\ref{GT}) is not necessarily achievable.
Although the expectation $\bE(\cdot)$ and the quantum relative entropy $\bD(\rho_t \| \wh{\rho}_t)$ in (\ref{GT}) depend on $\rho_t$, the actual density operator does not enter (\ref{QEFnom}). Therefore, maximization of both sides of (\ref{GT}) over the uncertainty class in (\ref{class}) leads to
\begin{align}
\nonumber
    \theta
  \sup_{\rho_t \in \fR_{t,\eps}}
  \bE \varphi(t)
  & \<
  \ln \wh{\Xi}_\theta(t)
  +
  \sup_{\rho_t \in \fR_{t,\eps}} \bD(\rho_t \| \wh{\rho}_t)\\
\label{GTsup}
  & \<
  \ln \wh{\Xi}_\theta(t)
  +
  \eps t,
\end{align}
which holds for all $\theta, t\>0$. By dividing both sides of (\ref{GTsup}) by $t>0$ and taking the infinite-horizon upper limit,  it follows that
\begin{align}
\nonumber
    \theta
  \limsup_{t\to +\infty}
  \Big(
    \frac{1}{t}
  \sup_{\rho_t \in \fR_{t,\eps}}
  \bE \varphi(t)
  \Big)
  & \<
  \eps
  +
  \limsup_{t\to +\infty}
  \Big(
    \frac{1}{t}
    \ln \wh{\Xi}_\theta(t)
  \Big)\\
\label{GTlimsup}
  & =
  \eps + \gamma(\theta),
\end{align}
where use is made of (\ref{gamma}). Now, division of both sides of (\ref{GTlimsup}) by $\theta>0$ yields the inequality
\begin{equation}
\label{GTlimsupdiv}
  \limsup_{t\to +\infty}
  \Big(
    \frac{1}{t}
  \sup_{\rho_t \in \fR_{t,\eps}}
  \bE \varphi(t)
  \Big)
   \<
  \frac{1}{\theta}(\eps + \gamma(\theta)),
\end{equation}
whose left-hand side is independent of $\theta$. Hence, minimization of the right-hand side of (\ref{GTlimsupdiv}) over $\theta>0$ establishes (\ref{supEphi}).
\end{proof}

For any given $\eps\>0$, the quantity $\frac{1}{\theta}(\eps + \gamma(\theta))$ under minimization  in (\ref{supEphi})  is a convex function of $\theta>0$, which is ensured by the convexity of each of the functions $\frac{\eps}{\theta}$ and $\frac{\gamma(\theta)}{\theta}$.
Also, since the right-hand side of (\ref{supEphi}) is increasing with respect to $\gamma(\theta)$, this inequality remains valid if $\gamma(\theta)$ is replaced with its upper bound. Such estimates for the asymptotic growth rate $\gamma(\theta)$ of the QEF are provided, for example, by \cite[Theorem 5]{VPJ_2017b} and will be used in the next section.

\section{AN UPPER BOUND FOR THE WORST-CASE QUADRATIC COST}
\label{sec:upper}

We will now consider the OQHO (\ref{Theta})--(\ref{AB}) with a Hurwitz matrix $A$. In the framework of the nominal model of Sections~\ref{sec:links} and \ref{sec:worst}, when the external fields are in the vacuum state, the stable OQHO has a unique invariant quantum state. This state is Gaussian \cite{KRP_2010} with zero mean and quantum covariance matrix $P+i\Theta$, where $P$ is the controllability Gramian of the pair $(A,B)$ 
satisfying the  algebraic Lyapunov equation
\begin{equation}
\label{P}
    AP+PA^\rT + BB^\rT = 0.
\end{equation}
If the OQHO is initialized at the invariant state $\rho_0$ in (\ref{rho0}), then the two-point quantum covariance matrix of the system variables is given by
\begin{equation}
\label{EXX}
    \wh{\bE}(X(s)X(t)^{\rT})
    =
    \Sigma(s-t) + i\Lambda(s-t),
    \qquad
    s,t\>0,
\end{equation}
where $\wh{\bE}$ is the expectation over the nominal system-field state as before, and the function
\begin{equation}
\label{Sigma}
    \Sigma(\tau)
    =
    \left\{
    \begin{matrix}
    \re^{\tau A}P & {\rm if}\ & \tau\> 0\\
    P \re^{-\tau A^{\rT}} & {\rm if}\ & \tau< 0\\
    \end{matrix}
    \right.
\end{equation}
is related to the matrix $P$ from (\ref{P}) in the same fashion as the function $\Lambda$ in (\ref{Lambda}) is related to the CCR matrix $\Theta$ from (\ref{Theta}). Since the matrix $A$ is Hurwitz, the appropriate matrix exponentials in (\ref{Lambda}) and (\ref{Sigma}) make  $\Sigma(\tau)+i\Lambda(\tau)$ in (\ref{EXX}) decay exponentially fast as $\tau\to \infty$. This property is inherited by the corresponding two-point covariance function $Z$ of the auxiliary quantum variables $\zeta_1, \ldots, \zeta_n$ in (\ref{psizeta}) and (\ref{zeta}):
\begin{equation*}
\label{EZZ}
    \wh{\bE}(\zeta(s)\zeta(t)^{\rT})
    =
    Z(s-t),
    \qquad
    s,t\>0,
\end{equation*}
where
\begin{equation}
\label{Z}
    Z(\tau)
    :=
    \sqrt{\Pi}(\Sigma(\tau) + i\Lambda(\tau))\sqrt{\Pi}.
\end{equation}
The exponential decay of $Z$ at infinity plays an important role in the following theorem, which specifies the robustness estimate of Theorem~\ref{th:worst} and is based substantially on \cite[Theorems 5 and 6]{VPJ_2017b}.

\begin{thm}
\label{th:devup}
Suppose the OQHO (\ref{dX}) has a Hurwitz matrix $A$ in (\ref{AB}), and the nominal system-field state is described by (\ref{rho0}), where $\rho_0$ is the invariant Gaussian state with zero mean and the real part $P$ of the quantum covariance matrix found from (\ref{P}).
Then the growth rate for the worst-case value of the quadratic cost in Theorem~\ref{th:worst} admits an upper bound
\begin{equation}
\label{supEphi1}
  \limsup_{t\to+\infty}
  \Big(
  \frac{1}{t}
  \sup_{\rho_t \in \fR_{t,\eps}}
  \bE \varphi(t)
  \Big)
  \<
    n\alpha
    \Big(1+\sigma + \sqrt{\sigma(2+\sigma)}\Big).
\end{equation}
%
Here,
\begin{align}
\label{sig}
    \sigma
    & :=
    \frac{2\eps}{n\mu},\\
\label{alpha}
    \alpha
    & :=
    \|\sqrt{\Pi} \sqrt{\Gamma}\|
    \|\Gamma^{-1/2}(P+i\Theta)\sqrt{\Pi}\|,
\end{align}
where $\|\cdot\|$ denotes the operator norm of a matrix, and $(\mu,\Gamma)$ is any pair of a scalar $\mu>0$ and a real positive definite symmetric matrix $\Gamma$ of order $n$  satisfying the algebraic Lyapunov inequality
\begin{equation}
\label{ALI}
    A\Gamma + \Gamma A^{\rT}\preccurlyeq - 2\mu \Gamma.
\end{equation}
\hfill$\square$
\end{thm}
\begin{proof}
As shown in the proof of \cite[Theorem 6]{VPJ_2017b},  any pair $(\mu,\Gamma)$ satisfying (\ref{ALI}) (which exists since $A$ is Hurwitz) and the corresponding parameter $\alpha$ in (\ref{alpha}) allow the operator norm of the quantum covariance matrix (\ref{Z}) to be bounded as
\begin{equation}
\label{Nup1}
    \|Z(\tau)\|
    \<
    \alpha \re^{-\mu |\tau|},
    \qquad
    \tau \in \mR.
\end{equation}
By applying \cite[Theorem 5]{VPJ_2017b}, it follows  that
the upper growth rate $\gamma(\theta)$ in (\ref{gamma})  for the nominal QEF $\wh{\Xi}_\theta$ in (\ref{QEFnom}) admits the bound
\begin{align}
\nonumber
    \gamma(\theta)
    & \<
    -\frac{n}{4\pi}
    \int_{-\infty}^{+\infty}
    \ln
    (1 - 2\theta F(\lambda))
    \rd
    \lambda\\
\label{upQEF}
    & =
            \frac{n}{2}
    \big(
        \mu - \sqrt{\mu^2 - 4\theta \alpha \mu}
    \big),
\end{align}
for any $\theta$ satisfying
\begin{equation}
\label{thetarange1}
    0 \< \theta < \frac{\mu}{4\alpha},
\end{equation}
where use is made of the Fourier transform of the right-hand side of (\ref{Nup1}):
$$    F(\lambda)
    :=
    \alpha
    \int_{-\infty}^{+\infty}
    \re^{-\mu |\tau|-i\lambda \tau}
    \rd \tau
    =
    \frac{2\alpha \mu}{\lambda^2 + \mu^2}.
$$
A combination of (\ref{supEphi}) of Theorem~\ref{th:worst} with (\ref{upQEF}) and (\ref{thetarange1}) leads to
\begin{align}
\nonumber
  \!\!\!\limsup_{t\to+\infty}
  \Big(
  \frac{1}{t}
  \sup_{\rho_t \in \fR_{t,\eps}}
  \bE \varphi(t)
  \Big)
   \<&\!\!
  \inf_{0 < \theta < \frac{\mu}{4\alpha}}\!\!\!
    \frac{\eps +             \frac{n}{2}
    \big(
        \mu - \sqrt{\mu^2 - 4\theta \alpha \mu}
    \big)}{\theta}
\\
\label{supEphi2}
  =&
  2n\alpha
  \inf_{0 < u< 1}
    \frac{
    \sigma
     +
        1 - \sqrt{1 - u}
    }{u},\!
\end{align}
where $\sigma$ is the dimensionless parameter in (\ref{sig}). The minimization on the right-hand side of (\ref{supEphi2}) leads to (\ref{supEphi1}), with the minimum being achieved at
\begin{equation}
\label{min}
    \theta
    =
    \frac{\mu}{2\alpha}
    \Big(
        1+\sigma-\sqrt{\sigma(2+\sigma)}
    \Big)
    \sqrt{\sigma(2+\sigma)}
\end{equation}
(which belongs to the interval in (\ref{thetarange1}) for any $\sigma\>0$).
\end{proof}

The inequality (\ref{ALI}) is guaranteed to have a positive definite solution $\Gamma$ for any $\mu<-\max \Re \fS$, where $\fS$ denotes the spectrum of the Hurwitz matrix $A$. A particular choice of such a matrix $\Gamma$ (which can employ the eigenbasis of $A$ in the case of diagonalizability) influences the parameter $\alpha$ in (\ref{alpha}) which enters (\ref{supEphi1}) together with $\mu$. The relation (\ref{min}) suggests an ``optimal'' value for the risk-sensitivity parameter $\theta$ in regard to the guaranteed robust performance bounds (\ref{supEphi1}) for the quadratic cost. The right-hand side of (\ref{min}) is strictly increasing with respect to the quantum relative entropy rate threshold $\eps\>0$ which specifies the quantum statistical uncertainty class (\ref{class}). The corresponding value $\theta=\theta_\eps$ vanishes at $\eps=0$ (when there is no uncertainty) and approaches the limit $\frac{\mu}{4\alpha}$ in (\ref{thetarange1}) as $\eps\to +\infty$ (large uncertainties).

\section{CONCLUSION}
\label{sec:conc}

We have discussed the links between the original risk-sensitive performance criterion for quantum control systems and its recent quadratic-exponential version. The quadratic-exponential cost functional has a bearing on robustness with respect to a class of uncertain quantum states of the system and its environment whose deviation from a nominal state is described in terms of the 
quantum relative entropy. These relations are similar to the robustness properties of the minimax LQG control for classical systems. In obtaining these results, we have employed algebraic
and quantum probabilistic 
techniques using, in particular, operator square roots and density operator transformations.
The findings of the paper can be of use in providing a rational choice of the risk-sensitivity parameter in the context of the entropy theoretic quantification of statistical uncertainty in the system-environment state.


%


\vspace{-1mm}

\begin{thebibliography}{99}


\bibitem{AM_1989}
B.D.O.Anderson, and J.B.Moore,
\emph{Optimal Control: Linear Quadratic Methods},
Prentice Hall, London, 1989.
\bibitem{B_1983}
V.P.Belavkin, On the theory of controlling observable quantum
systems, \emph{Autom. Rem. Contr.}, vol. 44, no. 2, 1983, pp. 178--188.
\bibitem{B_2010}
V.P.Belavkin, Noncommutative dynamics and generalized master equations, \emph{Math. Notes}, vol. 87, no. 5, 2010, pp. 636--653.
\bibitem{BV_1985}
A.Bensoussan, and J.H.van Schuppen, Optimal control of
partially observable stochastic systems with an
exponential-of-integral performance index, \textit{SIAM J. Control
Optim.}, vol. 23, 1985, pp. 599--613.

\bibitem{BH_2006}
L.Bouten, and R.van Handel, On the separation principle of quantum control,
arXiv:math-ph/0511021v2, August 22, 2006.
\bibitem{BVJ_2007}
L.Bouten, R.Van Handel, M.R.James,
An introduction to quantum filtering,
\emph{SIAM J. Control Optim.}, vol. 46, no. 6, 2007, pp. 2199--2241.



\bibitem{CT_2006}
T.M.Cover, and J.A.Thomas, \textit{Elements of Information Theory},
Wiley, New York, 2006.

\bibitem{CH_1971}
C.D.Cushen, and R.L.Hudson, A quantum-mechanical central limit theorem,
\emph{J. Appl. Prob.}, vol. 8, no. 3, 1971, pp. 454--469.
\bibitem{DDJW_2006}
C.D'Helon, A.C.Doherty, M.R.James, and S.D.Wilson,
Quantum risk-sensitive control,
Proc. 45th IEEE CDC,
San Diego, CA, USA, December 13--15, 2006, pp. 3132--3137.



\bibitem{DE_1997}
P.Dupuis, and R.S.Ellis,
\textit{A Weak Convergence Approach to the Theory of Large Deviations},
Wiley, New York, 1997.
\bibitem{DJP_2000}
P.Dupuis, M.R.James, and I.R.Petersen, Robust properties of risk-sensitive control,
 \textit{Math. Control Signals Syst.}, vol.  13, 2000, pp. 318--332.
\bibitem{EB_2005}
S.C.Edwards, and V.P.Belavkin,
Optimal quantum filtering and
quantum feedback control,
arXiv:quant-ph/0506018v2, August 1,  2005.




\bibitem{GZ_2004}
C.W.Gardiner, and P.Zoller,
\textit{Quantum Noise}.
Springer, Berlin, 2004.

\bibitem{G_1965}
S.Golden, Lower bounds for the Helmholtz function, \textit{Phys. Rev.}, vol.
137, 1965, pp. B1127--B1128.
\bibitem{GJ_2009}
J.Gough, and M.R.James,
Quantum feedback networks: Hamiltonian
formulation,
\emph{Commun. Math. Phys.},  vol. 287, 2009, pp. 1109--1132.
\bibitem{H_1991}
A.S.Holevo, Quantum stochastic calculus,
\textit{J. Sov. Math.},
vol. 56, no. 5, 1991, pp. 2609--2624.

\bibitem{H_2001}
A.S.Holevo, \textit{Statistical Structure of Quantum Theory}, Springer, Berlin, 2001.
\bibitem{HP_1984}
R.L.Hudson,  and K.R.Parthasarathy,
Quantum Ito's formula and stochastic evolutions,
\textit{Commun. Math. Phys.}, vol.  93, 1984, pp. 301--323
\bibitem{J_1973}
D.H.Jacobson, Optimal stochastic linear systems with
exponential performance criteria and their relation to
deterministic differential games, \textit{IEEE Trans. Autom.
Control}, vol. 18, 1973,  pp. 124--31.
\bibitem{J_2004}
M.R.James, Risk-sensitive optimal control of quantum systems,
\emph{Phys. Rev. A},  vol. 69, 2004, pp. 032108-1--14.
\bibitem{J_2005}
M.R.James, A quantum Langevin formulation of risk-sensitive optimal control,
\emph{J. Opt. B}, vol. 7, 2005, pp. S198--S207.

\bibitem{JG_2010}
M.R.James, and J.E.Gough, Quantum dissipative systems and feedback control design by interconnection,
\textit{IEEE Trans.
Automat. Contr.}, vol. 55, no. 8, 2008, pp. 1806--1821.
\bibitem{JNP_2008}
M.R.James, H.I.Nurdin, and I.R.Petersen,
$H^{\infty}$ control of
linear quantum stochastic systems,
\textit{IEEE Trans.
Automat. Contr.}, vol. 53, no. 8, 2008, pp. 1787--1803.
%




\bibitem{KS_1972}
H.Kwakernaak, and R.Sivan,
\textit{Linear Optimal Control Systems},
Wiley, New York, 1972.
\bibitem{LL_1991}
L.D.Landau, and E.M.Lifshitz, \textit{Quantum Mechanics: Non-relativistic Theory},
3rd Ed., Pergamon Press,  Oxford, 1991.
\bibitem{MP_2009}
A.I.Maalouf, and I.R.Petersen,
Coherent LQG control for a class of linear complex quantum systems,
IEEE European Control Conference, Budapest, Hungary, 23-26 August 2009,
pp. 2271--2276.
%



\bibitem{M_1998}
E.Merzbacher,
\textit{Quantum Mechanics}, 3rd Ed.,
Wiley, New York, 1998.
\bibitem{M_1995}
P.-A.Meyer,
\textit{Quantum Probability for Probabilists},
Springer, Berlin, 1995.
\bibitem{NC_2000}
M.A.Nielsen, and I.L.Chuang,
\textit{Quantum Computation and Quantum Information},
Cambridge University Press, Cambridge, 2000.
\bibitem{NJP_2009}
H.I.Nurdin, M.R.James, and I.R.Petersen,
Coherent quantum LQG
control,
\textit{Automatica}, vol.  45, 2009, pp. 1837--1846.
%
\bibitem{OP_1993}
M.Ohya, and D.Petz, \textit{Quantum Entropy and Its Use}, Springer-Verlag, Berlin,
1993.
\bibitem{OW_2010}
M.Ohya, and N.Watanabe,
Quantum entropy and its applications to quantum
communication and statistical physics,
\textit{Entropy}, vol.  12, 2010, pp. 1194--1245.


\bibitem{P_1992}
K.R.Parthasarathy,
\textit{An Introduction to Quantum Stochastic Calculus},
Birk\-h\"{a}user, Basel, 1992.
\bibitem{KRP_2010}
K.R.Parthasarathy,
What is a Gaussian state?
\textit{Commun. Stoch. Anal.}, vol. 4, no. 2, 2010, pp. 143--160.

\bibitem{PS_1972}
K.R.Parthasarathy, and K.Schmidt,
\emph{Positive Definite Kernels, Continuous Tensor Products, and Central Limit Theorems of Probability Theory},
Springer-Verlag, Berlin, 1972.

\bibitem{P_2015}
K.R.Parthasarathy,
Quantum stochastic calculus and quantum Gaussian processes,
\textit{Indian Journal of Pure and Applied Mathematics}, 2015, vol. 46, no. 6, 2015, pp. 781--807.
\bibitem{P_2006}
I.R.Petersen, Minimax LQG control, \textit{Int. J. Appl. Math. Comput. Sci.}, vol. 16, no. 3, 2006, pp. 309--323.
\bibitem{P_2014}
I.R.Petersen,
Guaranteed non-quadratic performance for quantum systems with nonlinear uncertainties,
 \textit{American Control Conference (ACC),
4-6 June 2014},  arXiv:1402.2086 [quant-ph], 10 February 2014.
\bibitem{P_2017}
I.R.Petersen,
Quantum linear systems theory,
\textit{Open Automat. Contr. Syst. J.},
vol. 8, 2017, pp. 67--93.


\bibitem{PJD_2000}
I.R.Petersen, M.R.James, and P.Dupuis, Minimax optimal control of stochastic uncertain
systems with relative entropy constraints, \textit{IEEE Trans. Automat. Contr.}, vol. 45, 2000,
pp. 398--412.
\bibitem{PUS_2000}
I.R.Petersen, V.A.Ugrinovskii, and A.V.Savkin, \textit{Robust Control Design Using $\cH^\infty$-
Methods}, Springer, London, 2000.
\bibitem{PUJ_2012}
I.R.Petersen, V.A.Ugrinovskii, and M.R.James,
Robust stability of uncertain linear
quantum systems,
\emph{Phil. Trans. R. Soc. A},  vol. 370, 2012, pp. 5354--5363.



%

\bibitem{S_1994}
J.J.Sakurai,
\textit{Modern Quantum Mechanics},
 Addison-Wesley, Reading, Mass., 1994.


\bibitem{SPJ_2007}
A.J.Shaiju, I.R.Petersen, and M.R.James,
Guaranteed cost LQG control of uncertain linear stochastic quantum systems,
American Control Conference, New York, 9–13 July 2007, pp. 2118--2123.


\bibitem{SP_2012}
A.J.Shaiju, and I.R.Petersen,
A frequency domain condition for the physical
realizability of linear quantum systems,
\emph{IEEE Trans. Automat. Contr.}, vol. 57, no. 8, 2012, pp. 2033--2044.

\bibitem{S_1996}
A.N.Shiryaev, \emph{Probability}, 2nd Ed., Springer, New York, 1996.


%

%

\bibitem{T_1965}
C.J.Thompson, Inequality with applications in statistical mechanics,
\textit{J. Math. Phys.}, vol. 6, 1965, pp. 1812--1813.
%

\bibitem{VP_2011b}
I.G.Vladimirov, and I.R.Petersen,
A dynamic programming approach to finite-horizon coherent quantum LQG control, 	
Australian Control Conference, Melbourne, 10--11 November, 2011, pp. 357--362,
arXiv:1105.1574v1 [quant-ph], 9 May 2011.

\bibitem{VP_2013a}
I.G.Vladimirov, and I.R.Petersen,
A quasi-separation principle and Newton-like scheme for coherent quantum LQG control,
\emph{Syst. Contr. Lett.}, vol. 62, no. 7, 2013, pp. 550--559.
\bibitem{VP_2013b}
I.G.Vladimirov, and I.R.Petersen,
Coherent quantum filtering for physically realizable linear quantum plants,
Proc. European Control Conference, IEEE, Zurich, Switzerland, 17-19 July 2013,  pp. 2717--2723.


\bibitem{VPJ_2017b}
I.G.Vladimirov, I.R.Petersen, and M.R.James
Multi-point Gaussian states, quadratic-exponential cost functionals, and large deviations estimates for linear quantum stochastic systems, {\tt arXiv:1707.09302 [math.OC], 28 July 2017}.


\bibitem{WM_2008}
D.F.Walls, and G.J.Milburn,
\emph{Quantum Optics}, 2nd Ed., Springer, Berlin, 2008.
\bibitem{W_1981}
P.Whittle, Risk-sensitive linear/quadratic/Gaussian
control, \textit{Adv. Appl. Probab.}, vol. 13,  1981, pp. 764--77.
\bibitem{WM_2010}
H.M.Wiseman, and G.J.Milburn,
\emph{Quantum measurement and control},
Cambridge University Press,
Cambridge.
\bibitem{YB_2009}
N.Yamamoto, and L.Bouten,
Quantum risk-sensitive estimation and robustness,
\emph{IEEE Trans. Automat. Contr.}, vol. 54, no. 1, 2009, pp. 92--107.

\bibitem{ZJ_2012}
G.Zhang, and M.R.James,
Quantum feedback networks and control: a brief survey,
\emph{Chinese Sci. Bull.}, vol. 57, no. 18, 2012, pp. 2200--2214,
arXiv:1201.6020v2 [quant-ph], 26 February 2012.
\end{thebibliography}
\end{document}